\documentclass{llncs}
\usepackage{fullpage}
\usepackage{graphicx}
\usepackage[center,scriptsize]{subfigure}
\usepackage{amssymb,amsmath}
\usepackage{upgreek}
\usepackage{psfrag}
\usepackage{theorem}
\usepackage{array}

\usepackage{algorithm,verbatim}

\usepackage[noend]{algorithmic}
\algsetup{linenosize=\scriptsize}

\usepackage{tikz}
\usetikzlibrary{decorations.pathmorphing}

\newcommand{\REM}[1]{}

\newtheorem{invariant}[theorem]{Invariant}

\newenvironment{appendix-theorem}[1]{\vspace{\theorempreskipamount}\noindent{\bf Theorem~#1~} \em }{\vspace{\theorempostskipamount}}
\newenvironment{appendix-lemma}[1]{\vspace{\theorempreskipamount}\noindent{\bf Lemma~#1~} \em }{\vspace{\theorempostskipamount}}
\newenvironment{appendix-corollary}[1]{\vspace{\theorempreskipamount}\noindent{\bf Corollary~#1~} \em }{\vspace{\theorempostskipamount}}
\newenvironment{appendix-invariant}[1]{\vspace{\theorempreskipamount}\noindent{\bf Invariant~#1~} \em }{\vspace{\theorempostskipamount}}

\newcommand {\vstar} {{\nu^*}}

\newcommand {\system}{tuple system }
\newcommand {\systemend}{tuple system}

\newcommand {\weight} {{\bf {w}}}
\newcommand {\maxM} {m^{*}}
\newcommand {\dict} {dict}
\newcommand {\MT} {Marked-Tuples }
\newcommand {\MTend} {Marked-Tuples}

\newcommand{\vone}{\vspace{.1in}}
\usepackage{url}
\urldef{\mailid}\path|{cavia,vlr}@cs.utexas.edu|

\title{Decremental All-Pairs ALL Shortest Paths \\ and Betweenness Centrality
\thanks{ This work was supported in part by NSF grant CCF-0830737.
The first author was also supported by CSE/14-15/824/NFIG/MEGA.
The second and third authors were also supported by NSF grant CCF-1320675.}
}
\author {Meghana Nasre \inst{1}, Matteo Pontecorvi \inst{2}, and Vijaya Ramachandran \inst{2}}
\institute{Indian Institute of Technology Madras, India \\
\email{meghana@cse.iitm.ac.in}
\and 
University of Texas at Austin, USA \\ \mailid
}

\begin{document}
\maketitle
\begin{abstract}
We consider the all pairs all shortest paths (APASP)
problem, which maintains the shortest path dag rooted at every
vertex in a directed graph $G=(V,E)$ with
positive edge weights.
For this problem  we present a
decremental algorithm (that supports the deletion of a vertex, or weight increases on edges
incident to a vertex).
Our algorithm  runs in amortized
 $O(\vstar^2 \cdot \log n)$ time per update, where
 $n = |V| $, and $\vstar$ 
bounds the number of edges 
that lie on shortest paths through any given vertex.
Our APASP algorithm can be used for the decremental computation of
betweenness centrality (BC),
a graph parameter that is widely used in the analysis of 
large complex networks.  No nontrivial decremental algorithm for either 
problem was known prior to our work.
Our method is a generalization of the decremental algorithm of Demetrescu and Italiano~\cite{DI04} for 
unique shortest paths,
and for graphs with 
$\vstar = O(n)$, we match the bound in~\cite{DI04}.
Thus for graphs with 
 a constant number of shortest paths between any pair of vertices, 
our algorithm maintains APASP and  BC scores in
amortized time
$O(n^2 \cdot \log n)$ under decremental updates, regardless of the number of edges in the graph.
\vspace{-0.1in}
\end{abstract}

\section{Introduction}

Given a directed graph $G=(V,E)$, with a positive real weight  $\weight(e)$ on each edge $e$,
we consider the problem of maintaining the shortest path 
dag rooted at every
vertex in $V$ (we will refer to these as the {\it SP dags}).
We use the term {\it all-pairs ALL shortest paths (APASP)}
to denote the collection of SP dags rooted at all $v\in V$, since
one can generate all the (up to exponential number of) shortest paths in
$G$ from these dags. These dags give a natural structural property of
$G$ which is of use in any application where several or all shortest
paths need to be examined. A particular application that motivated our
work is the computation of betweenness centrality (BC) scores of
vertices in a graph \cite{Freeman77}.

In this paper we present a decremental 
algorithm for the APASP problem, where each update in $G$
 either deletes 
  or increases the weight of some edges incident on a vertex.
Our method is a generalization of the method developed by
Demetrescu and Italiano~\cite{DI04}  (the `DI' method)
for decremental APSP where only one shortest path is needed.
The DI  decremental
algorithm~\cite{DI04} runs in $O(n^2 \cdot \log n)$ amortized
time per update, for a sufficiently long update sequence.
This decremental algorithm is also extended to a fully dynamic algorithm in~\cite{DI04} 
that runs in $O(n^2 \cdot \log^3 n)$ time, and this
result was improved to $O(n^2 \cdot \log^2 n)$ 
amortized time by Thorup~\cite{Thorup04};
both algorithms have within them essentially the same decremental algorithm.
We briefly discuss the
fully dynamic case, for which  APASP results have been obtained recently by two of the authors,
at the end of the paper.

In \cite{DI04,Thorup04} the goal was to compute all pairs shortest path distances, and hence these 
algorithms preprocess the graph in order to have a unique shortest
path between every pair of vertices. The unique shortest paths assumption, although not restrictive 
in their case, is crucial to the correctness and time complexity of their algorithms. We are interested
in the more general problem of APASP, and this poses several challenges in generalizing the approach
in~\cite{DI04}.

In addition to APASP, our method gives decremental algorithms for the following two problems.

\vone
\noindent
{\bf Locally Shortest Paths (LSPs).} 
For a path $\pi_{xy} \in G$, we define the $\pi_{xy}$ {\it distance}
 from $x$ to
$y$ as $\weight(\pi_{xy}) = \sum_{e \in \pi_{xy}}\weight(e)$, and the
$\pi_{xy}$  {\it length} from $x$ to $y$ as the number of edges 
on $\pi_{xy}$.
 For any $x, y \in V$, $d(x, y)$ denotes the shortest path distance from
$x$ to $y$ in $G$.
A path $\pi_{xy}$ in $G$ is a {\it locally shortest path (LSP)}~\cite{DI04}  if
either $\pi_{xy}$ contains a single vertex, or
every proper subpath of $\pi_{xy}$ is a shortest path in $G$.
As noted in \cite{DI04}, every shortest path
(SP) is an LSP, but an LSP need not be
an SP (e.g., every single edge is an LSP).

The DI method maintains all LSPs in a graph 
with unique shortest paths, and these are key to efficiently maintaining
shortest paths under decremental and fully dynamic updates. 
The decremental method we present here maintains 
all LSPs for all (multiple) shortest paths in a graph,
using a compact {\it tuple} representation.

\vone
\noindent
{\bf Betweenness Centrality (BC).}
Betweenness centrality is a widely-used measure
in the analysis of large complex networks, and is defined as follows.
For any pair $x, y$ in $V$, let $\sigma_{xy}$ denote the number of shortest paths
from $x$ to $y$ in $G$, and let
 $\sigma_{xy}(v)$ denote the number of shortest
paths from $x$ to $y$ in $G$ that pass through $v$. Then, $BC(v) = \sum_{s \neq v, t \neq v} \frac{\sigma_{st}(v)}{\sigma_{st}}$.
This measure is often used as an index that  determines the relative importance of
$v$ in the network.
Some applications of BC
include analyzing social interaction networks \cite{KA12},
identifying lethality in biological networks \cite{PCW05},
and identifying key actors in terrorist networks \cite{Coffman,Krebs02}.
Heuristics for dynamic betweenness centrality with good experimental
performance are given in~\cite{GreenMB12,Lee12,SinghGIS13},
but none of these algorithms
provably improve on the widely used static algorithm by
Brandes~\cite{Brandes01}, which runs in $O(mn + n^2 \cdot \log n)$ time
on any class of graphs, where $m=|E|$.

Recently, 
the authors gave a simple incremental BC algorithm \cite{NPR14}, that provably improves on Brandes' on sparse graphs, and
also typically improves on Brandes' in dense graphs (e.g., in the setting of
Theorem \ref{th:main2} below). In this paper, we complement the results in \cite{NPR14};
however, decremental updates are considerably more challenging (similar to APSP, as
noted in \cite{DI04}).

The key step in  the recent incremental BC
algorithm~\cite{NPR14} is the incremental maintenance of the APASP dags
(achieved  there using techniques unrelated to the current paper).
After the updated dags are obtained, the BC scores
can be computed in time linear in the combined sizes of the APASP dags
(plus $O(n^2)$).
Thus, if we instead use our decremental APASP algorithm in 
the key step in~\cite{NPR14}, we obtain
a decremental algorithm for 
BC with the same 
bound as APASP.

\vone
\noindent
{\bf Our Results.}
 Let $\vstar$ be the maximum
number of edges that lie on shortest paths through any given vertex in $G$;
thus, $\vstar$ also bounds the number of edges
that lie on any single-source shortest path dag.
Let $m^*$ be the number of edges in $G$ that lie on shortest paths
(see, e.g., Karger et al.~\cite{KKP93}). 
Our main result is the following theorem, where we have assumed
that $\vstar=\Omega (n)$.

\begin{theorem}\label{th:main}
Let $\Sigma$ be a sequence of decremental updates on $G=(V,E)$. Then,
all SP dags, all LSPs,  and all BC scores can be
maintained in amortized time $O( \vstar^2 \cdot \log {n})$ per update 
when $|\Sigma| = \Omega (m^*/\vstar)$.
\end{theorem}

\noindent
{\bf Discussion of the Parameters.}
As  noted in \cite{KKP93}, it
is well-known that $m^* = O(n \log{n})$
with high probability in a complete graph where edge weights are chosen from
a large class of probability distributions.
Since $\vstar \leq m^*$,
our algorithms will have an amortized bound of  $O(n^2 \cdot \log^3 n)$ on
such graphs. Also,
$\vstar =O(n)$ in any graph with only a constant number of shortest paths between every
pair of vertices, even though $m^*$ can be $\Theta (n^2)$ in the worst case
even in graphs with unique shortest paths.
In fact $\vstar = O(n)$ 
in some  graphs that
have an exponential number of shortest paths between
some pairs of vertices.
In all such cases, and more generally, when the number of edges on shortest paths through
any single vertex is $O(n)$, our algorithm
will run in amortized $O(n^2 \cdot \log n)$
time per decremental update.
Thus we have:
\begin{theorem}\label{th:main2}
Let $\Sigma$ be a sequence of decremental updates on graphs where the number
of edges on shortest paths through any single vertex is $O(n)$. 
Then, all SP dags, all  LSPs, and all  BC scores can be
maintained in amortized time $O( n^2 \cdot \log {n})$ per update
when $|\Sigma| = \Omega (m^*/n)$.
\end{theorem}

\begin{corollary}
If the number of shortest paths for any vertex pair is bounded by a constant, then
decremental APASP, LSPs, and BC have
 amortized cost $O(n^2 \cdot \log n)$
per update when the update sequence has length $\Omega (m^*/n)$.
\end{corollary}

\begin{figure}[htbp]
\begin{minipage}{0.25\linewidth}
\centering
\begin{tikzpicture}[every node/.style={circle, draw, inner sep=0pt, minimum width=5pt}]
\node (x1)[label=above:$x'$] at (0,1)  {};
\node (x)[label=left:$x$] at (0,0.2)  {};
\node (a1)[label=left:$a_1$] at (-1,-0.6) {};
\node (a2)[label=right:$a_2$] at (0,-0.6) {};
\node (a3)[label=right:$a_3$] at (1,-0.6) {};
\node (v)[label=right:$v$] at (0,-1.4)  {};
\node (v1)[label=left:$v_1$] at (-1.5,-1.4)  {};
\node (v2)[label=right:$v_2$] at (1.5,-1.4)  {};
\node (b1)[label=below left:$b_1$] at (-0.5,-2.2) {};
\node (b2)[label=below right:$b$] at (0.5,-2.2) {};
\node (y)[label=below:$y$] at (0.5,-3) {};
\node (y1)[label=below:$y_1$] at (-0.5,-3) {};
\path[every node/.style={font=\sffamily\small}]
    (a1) edge node [left] {\textbf{\scriptsize{2}}} (v1)
    (a1) edge node [right] {\textbf{\scriptsize{4}}} (b1)
    (v1) edge node [left] {\textbf{\scriptsize{2}}} (b1);
\draw[->] (x1) -- (x);
\draw[->] (x) -- (a1);
\draw[->] (x) -- (a2);
\draw[->] (x) -- (a3);
\draw[->] (a1) -- (v1) ;
\draw[->] (a1) -- (b1) ;
\draw[->] (a1) -- (v);
\draw[->] (a2) -- (v);
\draw[->] (a2) -- (v2);
\draw[->] (a3) -- (v2);
\draw[->] (v1) -- (b1);
\draw[->] (v) -- (b1);
\draw[->] (v) -- (b2);
\draw[->] (v2) -- (b2);
\draw[->] (b1) -- (y1);
\draw[->] (b2) -- (y);
\end{tikzpicture}
\caption{Graph $G$}
\label{fig:illust-example}
\end{minipage}
\begin{minipage}{0.75\linewidth}
\centering
\renewcommand{\arraystretch}{1.3}
\begin{tabular}{|c| c|}
     \hline
    {Set} & {$G$ (before update on $v$)} \\
 \hline \hline
$P(x,y)$ & ${\{((xa_1,by),4,1), ((xa_2,by),4,2), }$ \\
$= P^*(x,y)$& $((xa_3, by), 4, 1) \}$  \\
\hline
$P(x,b_1)$ & $\{(xa_1,vb_1),3,1),((xa_2,vb_1),3,1)\}$ \\
\hline
$P^*(x,b_1)$ & $\{((xa_1,vb_1),3,1),((xa_2,vb_1),3,1)\}$ \\
\hline
$L^*(v,y_1)$ & $\{a_1,a_2\}$ \\
\hline
$L(v,b_1y_1)$ & $\{a_1,a_2\}$ \\
\hline
$R^*(x,v)$ & $\{b,b_1\}$ \\
\hline
$R(xa_2,v)$ & $\{b,b_1\}$ \\
\hline
\end{tabular}
\caption{A subset of the tuple-system for $G$ in Fig.~\ref{fig:illust-example}}
\label{fig:example-table}
\end{minipage}
\end{figure}

\vone
\noindent{\bf The  DI method.}  
Here we will use an example to  
give a quick review of the DI approach~\cite{DI04}, which forms the basis
for our method.
Consider the graph $G$ in Fig.~\ref{fig:illust-example},
where all edges have weight 1 except for the ones with explicit weights. 

As in DI, let us assume here that
$G$ has been  pre-processed to identify
a unique shortest path between every pair of vertices. In $G$ 
 the shortest path from $a_1$ to $b_1$ is $\langle a_1, v, b_1\rangle$ and
has weight $2$, and by definition, the paths $p_1 = \langle a_1, b_1\rangle$
and $p_2 = \langle a_1, v_1, b_1 \rangle$ of weight $4$ are both LSPs. Now consider
a decremental update on $v$ that increases  $\weight (a_1, v)$ to 10 and $\weight (a_2,v)$ to 5, and let $G'$ be  the resulting graph (see Fig.~\ref{fig:illust-example2}). In $G'$
both  $p_1$ and $p_2$ become shortest paths.
Furthermore, a {\em left extension} of the path $p_1$, namely $p_3 = \langle x, a_1, b_1\rangle $ becomes
a shortest path from $x$ to $b_1$ in $G'$. Note that the path $p_3$ is not even an LSP
in the graph $G$; however, it is obtained as a left extension of a path 
that has become shortest after the update.

The elegant method
of storing LSPs and creating longer LSPs by left and right extending shortest paths is the basis of the DI approach~\cite{DI04}. 
To achieve this, the DI approach uses a succinct representation of SPs, LSPs and their left and right extensions 
using suitable data structures. It then
uses a procedure {\em cleanup} to remove from the data structures all the shortest paths and LSPs that contain the updated
vertex $v$, and a complementary
procedure {\em fixup} that first adds all the trivial LSPs (corresponding to edges incident on $v$),
and then restores the shortest paths and LSPs between all pairs of vertices.
The DI approach thus efficiently maintains a single shortest
path between all pairs of vertices under decremental updates. 

\vone
\noindent
{\bf Roadmap.}
In this paper we are interested in maintaining {\em all}  shortest paths 
for all vertex pairs and this requires several enhancements to the 
DI method~\cite{DI04}.
In Section~\ref{sec:tuple-system} we present  a  new {\it tuple system} which succinctly represents all LSPs in a graph with multiple
shortest paths
and in Section~\ref{sec:alg} we present our decremental algorithm for maintaining this tuple system,
and hence for maintaining APASP and BC scores.
\section{A System of Tuples}
\label{sec:tuple-system}
In this section we present an efficient representation of the set of
SPs and LSPs for an edge weighted graph $G = (V, E)$. We first define the
notions of {\em tuple} and {\em triple}.

\vone
\noindent
{\bf Tuple.}
A tuple, $\tau = (xa, by)$, represents the set of LSPs in $G$, all 
of which  use
the same first edge $(x,a)$ and the same last edge $(b,y)$. 
The weight of every path
represented by $\tau$ is $\weight(x, a)$ + $d(a, b) + \weight(b,y)$. 
We call $\tau$ a {\it locally shortest path tuple (LST)}.
In addition, if $d(x, y) = \weight(x, a) + d(a, b) + \weight(b, y)$, then 
$\tau$ is a {\it shortest path tuple (ST)}.
Fig.~\ref{fig:tau} shows a tuple $\tau$.

\vone
\noindent
{\bf Triple.}
A triple $\gamma=(\tau, wt, count)$, represents the tuple $\tau=(xa,by)$ that contains
$count > 0$ number of paths from $x$ to $y$, each with weight
$wt$. In Fig.~\ref{fig:illust-example}, the triple $((xa_2,by),4,2)$ represents two paths from $x$ to $y$, namely
$p_1 = \langle x, a_2, v, b, y \rangle$ and $p_2 = \langle x, a_2, v_2, b, y \rangle$ both having weight $4$.

\vone
\noindent
{\bf Storing Locally Shortest Paths.}
\label{sec:storing-LST}
We use triples to succinctly store all LSPs and SPs
for each vertex pair  in $G$. For $x, y \in V$, we define:
\begin{eqnarray*}
P(x,y) &=& \{\mbox{$((xa, by), wt, count)$: $(xa, by)$ is an LST from $x$ to $y$ in $G$} \} \\
P^{*}(x,y) &=& \{\mbox{$((xa, by), wt, count)$: $(xa, by)$ is an ST from $x$ to $y$ in $G$} \}.
\end{eqnarray*}
Note that all triples in $P^{*}(x,y)$ have the same weight.  We will
use the term LST to denote 
either a locally shortest tuple or a triple representing 
a set of LSPs,
and it will be clear from the context whether we mean a triple or a tuple.

\begin{figure}[ht]
\begin{minipage}{0.35\linewidth}
\centering
\begin{tikzpicture}[every node/.style={circle, draw, inner sep=0pt, minimum width=5pt}]
\node (x1)[label=above:$x'$] at (0,1)  {};
\node (x)[label=left:$x$] at (0,0.2)  {};
\node (a1)[label=left:$a_1$] at (-1,-0.6) {};
\node (a2)[label=right:$a_2$] at (0,-0.6) {};
\node (a3)[label=right:$a_3$] at (1,-0.6) {};
\node (v)[label=right:$v$] at (0,-1.4)  {};
\node (v1)[label=left:$v_1$] at (-1.5,-1.4)  {};
\node (v2)[label=right:$v_2$] at (1.5,-1.4)  {};
\node (b1)[label=below left:$b_1$] at (-0.5,-2.2) {};
\node (b2)[label=below right:$b$] at (0.5,-2.2) {};
\node (y)[label=below:$y$] at (0.5,-3) {};
\node (y1)[label=below:$y_1$] at (-0.5,-3) {};
\path[every node/.style={font=\sffamily\small}]
    (a1) edge node [left] {\textbf{\scriptsize{2}}} (v1)
    (a1) edge node [right] {\textbf{\scriptsize{4}}} (b1)
    (v1) edge node [left] {\textbf{\scriptsize{2}}} (b1)
    (a1) edge node [right] {\textbf{\scriptsize{10}}} (v)
    (a2) edge node [right] {\textbf{\scriptsize{5}}} (v);
\draw[->] (x1) -- (x);
\draw[->] (x) -- (a1);
\draw[->] (x) -- (a2);
\draw[->] (x) -- (a3);
\draw[->] (a1) -- (v1) ;
\draw[->] (a1) -- (b1) ;
\draw[->] (a1) -- (v);
\draw[->] (a2) -- (v);
\draw[->] (a2) -- (v2);
\draw[->] (a3) -- (v2);
\draw[->] (v1) -- (b1);
\draw[->] (v) -- (b1);
\draw[->] (v) -- (b2);
\draw[->] (v2) -- (b2);
\draw[->] (b1) -- (y1);
\draw[->] (b2) -- (y);
\end{tikzpicture}
\caption{Graph $G'$}
\label{fig:illust-example2}
\end{minipage}
\begin{minipage}{0.6\linewidth}

\centering
\renewcommand{\arraystretch}{1.3}
\begin{tabular}{|c| c|}
     \hline
    {Set}  &  {$G'$ (with $\weight(a_1, v) = 10$, $\weight (a_2, v) = 5$)} \\
\hline \hline
$P(x,y)$  & $\{((xa_2,by),4,1), ((xa_3,by),4,1)\}$ \\
$= P^*(x,y)$& \\
\hline
$P(x,b_1)$ &  $ \{((xa_1,v_1b_1),5,1),  ((xa_2,vb_1),7,1),$\\
&  $((xa_1,a_1b_1),5,1) \}$ \\
\hline
$P^*(x,b_1)$ &  $\{((xa_1,v_1b_1),5,1),((xa_1,a_1b_1),5,1)\}$ \\
\hline
$L^*(v,y_1)$ &   $\{a_2\}$ \\
\hline
$L(v,b_1y_1)$ &  $\{a_2\}$\\
\hline
$R^*(x,v)$ &   $\emptyset$ \\
\hline
$R(xa_2,v)$ &  $ \{b_1\}$ \\
\hline
\end{tabular}
\vspace{-0.08in}
\caption{A subset of the tuple-system for  $G'$}
\label{fig:example-table2}
\end{minipage}
\vspace{-0.2in}
\end{figure}

\vone
\noindent
{\bf Left Tuple and Right Tuple.}
A left tuple (or $\ell$-tuple), 
$\tau_{\ell} = (xa, y)$,
represents the
set of LSPs from $x$ to $y$, all of which use 
the same first edge $(x,a)$.
The weight of every path
represented by $\tau_{\ell}$ is $\weight(x, a)$ + $d(a, y)$. 
If $d(x, y) = \weight(x, a) + d(a, y)$, then $\tau_{\ell}$ represents the
set of shortest paths from $x$ to $y$,
all of which use the first edge $(x, a)$. 
A right tuple ($r$-tuple) $\tau_r = (x, by)$ is defined analogously.
Fig.~\ref{fig:taul} and 
Fig.~\ref{fig:taur} show a left tuple and a right tuple respectively.
In the following, we will say that a tuple (or $\ell$-tuple or $r$-tuple)
{\it contains} a vertex $v$, if at least one of the paths represented by the
tuple contains $v$.
For instance, in Fig.~\ref{fig:illust-example}, the tuple $(xa_2, by)$ 
contains the vertex $v$ as well as the vertex $v_2$.

\begin{figure}
\centering
\subfigure[tuple  $\tau = (xa,by)$]{
\makebox[.3\textwidth]{
\begin{tikzpicture}[every node/.style={circle, draw, inner sep=0pt, minimum width=5pt}]
\node (x)[label=above:$x$] at (0,0)  {};
\node (a)[label=left:$a$] at (0,-0.8) {};
\node (b)[label=below left:$b$] at (0,-2.2) {};
\node (y)[label=below:$y$] at (0,-3) {}; 
\draw[->] (x) -- (a);
\path[->,decoration={snake}] { (a) edge[decorate] (b)};
\path[->,decoration={snake}] { (a) edge[bend right=60,decorate] (b.west)};
\path[->,decoration={snake}] { (a) edge[bend left=60,decorate] (b.east)};
\draw[->] (b) -- (y);
\end{tikzpicture}} \label{fig:tau}} 
\subfigure[$\ell$-tuple  $\tau_{\ell} = (xa,y)$]{
\makebox[.3\textwidth]{
\begin{tikzpicture}[every node/.style={circle, draw, inner sep=0pt, minimum width=5pt}]
\node (x)[label=above:$x$] at (0,0)  {};
\node (a)[label=left:$a$] at (0,-1) {};
\node (y)[label=below:$y$] at (0,-3) {}; 
\draw[->] (x) -- (a);
\path[->,decoration={snake}] { (a) edge[decorate] (y)};
\path[->,decoration={snake}] { (a) edge[bend right=60,decorate] (y.west)};
\path[->,decoration={snake}] { (a) edge[bend left=60,decorate] (y.east)};
\end{tikzpicture}} \label{fig:taul}} 
\subfigure[$r$-tuple  $\tau_r = (x,by)$]{
\makebox[.3\textwidth]{
\begin{tikzpicture}[every node/.style={circle, draw, inner sep=0pt, minimum width=5pt}]
\node (x)[label=above:$x$] at (0,0)  {};
\node (b)[label=below left:$b$] at (0,-2) {};
\node (y)[label=below:$y$] at (0,-3) {}; 
\draw[->] (b) -- (y);
\path[->,decoration={snake}] { (x) edge[decorate] (b)};
\path[->,decoration={snake}] { (x) edge[bend right=60,decorate] (b.west)};
\path[->,decoration={snake}] { (x) edge[bend left=60,decorate] (b.east)};
\end{tikzpicture}} \label{fig:taur}}
\caption{Tuples}
\end{figure}
\vspace{-0.1in}

\vone
\noindent
{\bf ST and LST Extensions.}
For a shortest path $r$-tuple $\tau_r = (x, by)$, we
define $L({\tau_r})$ to be 
the set of vertices which can be used as pre-extensions to create  LSTs
in~$G$.
Similarly, for a shortest path $\ell$-tuple $\tau_{\ell} = (xa, y)$,
$R(\tau_{\ell})$ is the set of vertices which can be used as post-extensions to
create LSTs in $G$. We do not 
define $R(\tau_r)$ and $L(\tau_{\ell})$. 
So we have:
\begin{eqnarray*}
L(x, by) &=& \{x': \mbox {$(x', x) \in E(G)$ and $(x'x, by)$ is an LST in $G$}\} \\
R(xa, y) &=& \{y': \mbox {$(y, y') \in E(G)$ and $(xa, yy')$ is an LST in $G$}\}.
\end{eqnarray*}

For $x, y \in V$, $L^{*}(x, y)$ denotes the set of vertices which can 
be used as pre-extensions to create shortest path tuples in $G$;
$R^{*}(x, y)$ is defined symmetrically:
\begin{eqnarray*}
L^{*}(x, y) &=& \{x': \mbox {$(x', x) \in E(G)$ and $(x'x, y)$ is a $\ell$-tuple representing SPs in $G$}\} \\
R^{*}(x, y) &=& \{y': \mbox {$(y, y') \in E(G)$ and $(x, yy')$ is an $r$-tuple representing SPs in $G$}\}.
\end{eqnarray*}

\noindent
Fig.~\ref{fig:example-table} shows a subset of these sets for the graph $G$ in Fig.~\ref{fig:illust-example}.

\paragraph{\bf Key Deviations from DI \cite{DI04}.} 
The assumption of unique shortest paths in \cite{DI04} 
ensures that $\tau = (xa, by)$,
$\tau_{\ell} = (xa, y)$, and $\tau_r = (x, by)$ all represent exactly the
same (single)  locally shortest path.
However, in our case, the set of paths represented by 
$\tau_{\ell}$ and $\tau_r$ can be different, and 
$\tau$ is a subset of paths represented by $\tau_{\ell}$ and
$\tau_r$. 
Our definitions of ST and LST extensions are
derived from the analogous definitions in \cite{DI04} for SP and LSP
extensions of paths.
For a path $\pi = x \rightarrow a \leadsto b \rightarrow y$,
DI defines sets $L$, $L^{*}$, $R$ and $R^{*}$.
In our case,
the analog of a path $\pi = x \rightarrow a \leadsto b \rightarrow y$ 
is a tuple $\tau = (xa, by)$, 
but to obtain efficiency, we  define
the set $L$ only for an $r$-tuple and the set $R$ only for an $\ell$-tuple. 
Furthermore, we define  $L^{*}$ and $R^{*}$ for each pair of vertices. 

In the following two lemmas we bound the total number of tuples in the graph and the total
number of tuples that contain a given vertex $v$.
These bounds also apply to the number of triples since there is exactly one triple for each tuple
in our \systemend.

\begin{lemma}
\label{lem:total-locally-shortest}
The number of LSTs in $G=(V,E)$ is bounded by $O(m^* \cdot \vstar)$.
\end{lemma}
\begin{proof}
For any LST $(\times a, \times \times)$, for some $a\in V$, 
the first  and last edge 
of any such tuple must lie on a shortest path containing $a$.
Let $E_a^{*}$ denote the set of edges that lie on shortest 
paths through $a$, and let $I_a$ be the set of incoming edges to $a$. Then,
there are at most $\vstar$ ways of choosing the last edge in
$(\times a, \times \times)$ and at most $E_a^{*} \cap I_a$ ways of
choosing the first edge in $(\times a, \times \times)$. Since
$\sum_{a \in V} |E_a^{*} \cap I_a | = m^*$, the 
number of LSTs in $G$ is at most
$\sum_{a \in V} \vstar \cdot |E_a^{*} \cap I_a | \leq m^* \cdot \vstar$.
{\hfill \qed} 
\end{proof}

\begin{lemma}
\label{lem:bound-tuples-thru-v}
The number of LSTs that contain a vertex $v$  
is $O({\vstar}^2)$.
\end{lemma}

\begin{proof}
We distinguish three different cases:

1. Tuples starting with $v$: 
for a tuple that starts with edge $(v,a)$, the last edge must lie on $a$'s
SP dag, so there are at most
$\vstar$ choices for the last edge. 
Hence, 
the number of tuples with $v$ as start vertex is 
at most $\sum_{a \in V \setminus {v}} \vstar \le n \cdot \vstar$.

2. Similarly, the number of tuples with $v$ as end vertex is at most $n \cdot \vstar$.

3. For any tuple $ \tau = (xa, by)$ that contains $v$ as an internal vertex,
both $(x,a)$ and $(b,y)$ lie on a shortest path through $v$, 
hence  the number of such tuples 
is at most $\vstar^2$.
{\hfill \qed}
\end{proof}

\section{Decremental Algorithm}\label{sec:alg}

Here we present  our decremental APASP algorithm. Recall that
a decremental update on a vertex $v$ either deletes or
increases the weights of a subset of edges incident on $v$\nopagebreak. We begin with the  data structures 
we use.

\vone
\noindent {\bf Data Structures.}
\label{subsec:static-impl}
For every $x, y$, $x \neq y$ in $V$, we maintain the following:
\begin{enumerate}
\item $P(x,y)$ -- a priority queue containing LSTs from $x$ to $y$ with weight as key.
\item $P^*(x,y)$ -- a priority queue containing STs from $x$ to $y$ with weight as key.
\item $L^*(x,y)$ -- a balanced search tree containing vertices with vertex ID as key.
\item $R^*(x,y)$ -- a balanced search tree containing vertices with vertex ID as key.
\end{enumerate}

For every $\ell$-tuple we have its right extension, and for every $r$-tuple its left extension.
These sets are stored as balanced search trees (BSTs)  with  the vertex ID as a key.
Additionally, we maintain all  tuples in a BST  $\dict$, with a
tuple $\tau=(xa, by)$ having key $[x, y, a, b]$. We also maintain  pointers from $\tau$ to
$R(xa, y)$ and  $L(x, by)$,  and to the corresponding triple containing $\tau$ in $P(x,y)$,
(and in $P^*(x,y)$ if $(xa, by)$ is an ST).
Finally, we maintain a sub-dictionary of $dict$  called \MT (explained below).
\MTend, unlike  the other data structures,
is specific only to one update. 

\vone
\noindent {\bf The Decremental Algorithm.}
Given  the updated vertex $v$ and the updated 
weight function $\weight'$ over all the incoming and outgoing edges of $v$, the decremental
algorithm performs two main steps \emph{cleanup} and \emph{fixup}, as in DI.
The cleanup procedure removes from the \system every LSP that contains the updated vertex
$v$. The following definition of a {\em new} LSP is from DI~\cite{DI04}.
\begin{definition}
\label{def:newLSP}
A path that is shortest (locally shortest) after an update to vertex $v$ is {\em new} if either
it was not an SP (LSP) before the update, or it contains $v$.
\end{definition}
The fixup procedure adds to the \system all the {\em new} shortest and locally shortest paths.
In contrast to DI,  recall that we store locally shortest paths in $P$ and $P^{*}$ as triples. 
Hence removing or adding paths implies decrementing or incrementing
the count in the relevant triple; thus a triple is removed or added only if its count goes down to zero or up from zero.
Moreover,  new tuples may be created through combining several existing tuples.
Some of the updated data structures for the graph $G'$ in Fig.~\ref{fig:illust-example2}, 
obtained
after a decremental update on $v$ 
in the graph $G$ in Fig.~\ref{fig:illust-example}, are schematized in 
Fig. \ref{fig:example-table2}.

\subsection{The Cleanup Procedure}\label{sec:cleanup}

Algorithm~\ref{algo:cleanup} (cleanup) uses
an initially empty  heap $H_c$ of
triples. It also initializes the empty dictionary \MTend.
The algorithm then creates
the trivial triple corresponding to the vertex $v$ and adds it to $H_c$ (Step~\ref{cleanup:init}, Algorithm~\ref{algo:cleanup}).
For a triple $((xa, by), wt, count)$ the key in $H_c$ is $[wt, x, y]$.
The algorithm repeatedly extracts
min-key triples from $H_c$ (Step~\ref{cleanup:extract-set}, Algorithm~\ref{algo:cleanup}) and {\em processes} them. The processing
of triples involves left-extending (Steps~\ref{process-cleanup:left-extend-start}--\ref{process-cleanup:left-extend-end}, Algorithm~\ref{algo:cleanup}) and right-extending triples (Step~\ref{cleanup:rextend}, Algorithm~\ref{algo:cleanup}) and removing from the tuple system
the set of LSPs thus formed. This is similar to cleanup in DI. However, 
since we deal with a set of paths instead of a single path, we need significant modifications, of
which we now highlight two: (i) Accumulation used in Step~\ref{cleanup:extract-set} and (ii) use of
\MT in Step~\ref{process-cleanup:left-extend-unmarked} and Step~\ref{process-cleanup:left-extend-markR}.

\vspace{-.1in}
\begin{algorithm}[h]
\scriptsize
\begin{algorithmic}[1]
\STATE $H_c \leftarrow \emptyset$; \MT $\leftarrow \emptyset$ \label{cleanup:initHcMT}
\STATE $\gamma \leftarrow ((vv,vv), 0, 1)$; add $\gamma$ to $H_c$ \label{cleanup:init}
\WHILE {$H_c \neq \emptyset$ } \label{cleanup:while}
        \STATE extract in $S$ all the triples with min-key $[wt,x,y]$ from $H_c$ \label{cleanup:extract-set}
\FOR {every $b$ such that $(x\times,by) \in S$}\label{process-cleanup:left-extend-start}
\STATE let $fcount'  = \sum_{i} ct_i$ such that $((xa_i, by), wt, ct_i ) \in S$ \label{process-cleanup:left-extend-fcount}
\FOR {every $x' \in L(x,by)$ such that $(x'x, by) \notin $ \MT } \label{process-cleanup:left-extend-unmarked}
                \STATE $wt' \leftarrow wt+\weight(x',x)$; $\gamma' \leftarrow ((x'x,by),wt', fcount')$; add $\gamma'$ to $H_c$ \label{process-cleanup:create-left}
                \STATE remove $\gamma'$ in $P(x',y)$ // decrements $count$ by $fcount$  \label{process-cleanup:left-extend-removeP}

                \IF {a triple for $(x'x, by)$ exists in  $P(x',y)$}
                        \STATE insert  $(x'x,by)$ in \MT \label{process-cleanup:left-extend-markR}
                \ELSE
                        \STATE delete $x'$ from $L(x,by)$ and delete $y$ from $R(x'x,b)$ \label{process-cleanup:left-extend-removeL}
                \ENDIF
                \IF {a triple for $(x'x, by)$ exists in  $P^*(x',y)$}  \label{process-cleanup:left-extend-checkSP}

                        \STATE remove $\gamma'$ in $P^*(x',y)$  // decrements $count$ by $fcount$ \label{process-cleanup:left-extend-removePstar}
					    \STATE \textbf{if} $P^*(x,y) = \emptyset$ \textbf{then} delete $x'$ from $L^{*}(x, y)$ \label{process-cleanup:left-extend-removeLstar}
					    \STATE \textbf{if} $P^*(x',b) = \emptyset$ \textbf{then} delete $y$ from $R^{*}(x',b)$ \label{process-cleanup:left-extend-removeRstar}
                \ENDIF
\ENDFOR
\ENDFOR \label{process-cleanup:left-extend-end}   
\STATE perform symmetric steps \ref{process-cleanup:left-extend-start} -- \ref{process-cleanup:left-extend-end} for right extensions  \label{cleanup:rextend}
\ENDWHILE
\end{algorithmic}
\caption{cleanup$(v)$}
\label{algo:cleanup}
\end{algorithm}

\vspace{-0.2in}
\paragraph{\bf Accumulation.} In Step~\ref{cleanup:extract-set} we extract a collection $S$ of triples all
with key $[wt,x,y]$
from $H_c$ and process them together in that iteration of the while loop. Assume
that for a fixed last edge $(b,y)$,
$S$ contains triples of the form $(xa_t, by)$, for $t = 1, \ldots, k $.  Our algorithm processes and left-extends all
these triples with the same last edge together. This ensures that, for any $x' \in L(x, by)$, we generate the triple $(x'x, by)$ exactly once. The accumulation is correct because
 any valid left extension for a triple $(xa_i, by)$
is also a valid left extension for $(xa_j, by)$ when both triples have the same weight.

\paragraph{\bf  \MTend.} 
The dictionary of \MT is used to ensure
that every path through the vertex $v$ is removed from the tuple system exactly once
and therefore counts of paths in triples are correctly maintained. Note that a path
of the form $(xa, by)$ can be generated either as a left extension of $(a, by)$ or
by a right extension of $(xa, b)$. This is true in DI as well. However, due to the assumption
of unique shortest paths they do not 
need to maintain counts of paths, and hence do not require the book-keeping using \MTend. 

\vone
\noindent
See sections \ref{app:accum} and \ref{app:marked} in the Appendix for more details.
 
\subsubsection{ Correctness and Complexity.}

We establish the correctness of cleanup in Lemma~\ref{lem:cleanup-corr2} and an upper bound on its
worst case time in Lemma~\ref{lem:cleanup-time}.

\begin{lemma}
\label{lem:cleanup-corr2}
After Algorithm~\ref{algo:cleanup}
is executed, the
counts of triples in $P$ ($P^{*}$) represent counts of LSPs (SPs)  in $G$
that do not pass through $v$. Moreover, the sets $L, L^{*}, R, R^{*}$ are correctly maintained.
\end{lemma}

\begin{proof}
To prove the lemma statement we show that the while loop in Step~\ref{cleanup:while}
of Algorithm~\ref{algo:cleanup} maintains the following invariants.\\
{\bf Loop Invariant:} At the start of each iteration of the while loop in Step~\ref{cleanup:while}
of Algorithm~\ref{algo:cleanup}, assume that the min-key triple to be extracted and processed from $H_c$ has key $[wt, x, y]$. Then
the following properties hold about the tuple system and $H_c$.
We assert the invariants about the sets $P$, $L$, and $R$. Similar arguments can be used to establish the
correctness of the sets $P^{*}$, $L^{*}$, and $R^{*}$.

\begin{enumerate}
\item [$\mathcal{I}_1$]\label{proof:item1c} For any $a, b \in V$, if $G$ contains $c_{ab}$ number of locally shortest paths of weight ${wt}$ of the form $(xa, by)$
passing through $v$,
then $H_c$ contains a triple $\gamma = ((xa, by), { wt}, c_{ab})$. Further, $c_{ab}$ has been
decremented from the initial count in the triple for $(xa, by)$  in $P(x,y)$.
\item [$\mathcal{I}_2$] Let $[\hat{wt}, \hat{x}, \hat{y}]$ be the key extracted from $H_c$ and processed in the previous iteration.
For any key $[wt_1, x_1, y_1] \leq [\hat{wt}, \hat{x}, \hat{y}]$, let
$G$ contain ${ c}  > 0 $ number of LSPs of weight  ${ wt_1}$ of the
form $(x_1a_1, b_1y_1)$. Further, let ${ c_v}$ (resp. ${ c_{\bar v}}$) denote the number of such LSPs
 that pass through $v$ (resp. do not pass through $v$).
Here ${ c_v + c_{\bar v} = c}$. Then,
\begin{enumerate}

\item \label{proof:item2c} if $c>c_v$ there is a triple in $P(x_1,y_1)$ of the form $(x_1a_1, b_1y_1)$  and weight $wt_1$ representing $c-c_v$ LSPs. If $c=c_v$ there is no such triple in $P(x_1,y_1)$.

\item \label{proof:item3c} $x_1 \in L(a_1, b_1y_1)$,  $y_1 \in R(x_1a_1, b_1)$, and    $(x_1a_1, b_1y_1) \in $ \MT
iff ${ {c_{\bar v}} > 0}$.
\item \label{proof:item4c} For every $x' \in L(x_1, b_1y_1)$, a triple corresponding to $(x' x_1, b_1y_1)$
with weight $wt'~=~wt_1~+~\weight (x',x_1)$ and the appropriate count is 
in $H_c$ if $[wt', x', y_1] \geq [wt, x, y]$. A similar claim can be stated for every $y' \in R(x_1a_1, y_1)$.
\end{enumerate}
\item [$\mathcal{I}_3$] \label{proof:item5c} For any key $[wt_2, x_2, y_2 ] \geq [wt, x, y]$, let
$G$ contain $c > 0$ LSPs of weight  ${ wt_2}$ of the
form $(x_2a_2, b_2y_2)$.
Further, let ${ c_v}$ (resp. ${ c_{\bar v}}$) denote the number of such LSPs
 that pass through $v$ (resp. do not pass through $v$).
Here ${ c_v + c_{\bar v} = c}$.
Then the tuple $(x_2 a_2, b_2 y_2) \in$ \MTend, iff $c_{\bar v} > 0$ and a triple for
$(x_2a_2, b_2 y_2)$ representing $c_v$ LSPs is present in $H_c$.
\end{enumerate}

\noindent
See section \ref{app:corr1} in the Appendix for the details of the proof.
\qed
\end{proof}

\begin{lemma}
\label{lem:cleanup-time}
For an update on a vertex $v$, Algorithm~\ref{algo:cleanup} takes $O( \vstar^2 \cdot \log{n})$ time.
\end{lemma}
\begin{proof}
The cleanup procedure examines a triple $\gamma$ only if the tuple in $\gamma$ contains the
updated vertex $v$. It removes each such  $\gamma$ from a constant number of data structures
 ($P, P^{*}, L, L^{*}, R, R^{*}$),
each with an $O(\log n)$ cost. In addition, each triple is inserted into $H_c$ and
extracted from it exactly once. Since the number of tuples containing $v$ is bounded by Lemma~\ref{lem:bound-tuples-thru-v},
the lemma follows.
\qed

\end{proof}

\vspace{-0.3in}
\subsection{The Fixup Procedure}
\label{sec:fixup}
The goal of the fixup procedure is to add to the \system all {\em new} shortest and locally shortest paths (recall Definition~\ref{def:newLSP}).

The fixup procedure (pseudo-code in Algorithm~\ref{algo:fixup})  works with a heap of triples ($H_f$ here),
which is initialized with a {\em candidate} shortest path triple for each pair of vertices.
Recall that for a pair $x,y$, there may be several triples of a given weight $wt$ in $P(x,y)$. Instead of
inserting all min-weight triples (which are candidates for shortest path
triples), our algorithm inserts exactly one
triple for every pair of vertices into $H_f$. This ensures that the number of triples examined during fixup is not too large.
Once $H_f$ is suitably initialized,
the  fixup algorithm
repeatedly extracts the set of triples with minimum key and processes them.
The main invariant for
the algorithm (similar to DI~\cite{DI04}) is that for a pair $x, y$,
the weight of the first set of triples
extracted from $H_f$ gives the distance from $x$ to $y$ in the updated
graph. Thus, these triples are all identified as shortest path triples,
and we need to extend them if in fact they represent {\em new}
shortest paths.
To readily identify triples containing paths through $v$ we 
use some additional book-keeping:
for every triple $\gamma$ we store the update number (update-num($\gamma$))
and a count of the number of paths in that triple that
pass through $v$ ($paths(\gamma,v)$). 
Finally, similar to cleanup, the fixup procedure also left and right extends triples to create
triples representing new locally shortest paths. 

We now describe the steps of the algorithm.
 
 Algorithm~\ref{algo:fixup} initializes $H_f$ in Steps~\ref{fixup:init1}--\ref{fixup:init2} as follows.
(i) For every edge incident on $v$, it creates a trivial triple $\gamma$ which is inserted into $H_f$ and $P$. It also sets $\mbox{update-num}(\gamma)$ and
$\mbox{paths}(\gamma, v)$ for each such $\gamma$;
(ii) For every $x, y \in V$, it adds a candidate min-weight triple from $P(x,y)$
to $H_f$ (even if
$P(x,y)$ contains several min-weight triples; this is done for efficiency).

Algorithm~\ref{algo:fixup}  executes Steps~\ref{fixup:main1}--\ref{fixup:phase3-main-check-end}
when for a pair $x, y$, the first  set of triples $S'$, all of weight $wt$, are extracted
from $H_f$. We claim  (Invariant~\ref{inv:fixup1}) that $wt$ denotes the 
shortest path distance from $x$ to $y$ in the updated graph. The goal of
Steps~\ref{fixup:main1}--\ref{fixup:phase3-main-check-end} is to create a set $S$
of triples that represent  {\em new} shortest paths, and this step is considerably more involved
than the corresponding step in DI.
In DI~\cite{DI04}, only a single path $p$ is extracted from $H_f$ possibly resulting in a {\em new} shortest path from $x$ to $y$. 
If $p$ is {\em new} then it is added to $P^*$ and the algorithm extends it to create new LSP. In our case, we extract not just multiple paths but multiple shortest path triples from $x$ to $y$, and 
some of these triples may not be in $H_f$. We now describe how
the algorithm generates the new shortest paths
in Steps~\ref{fixup:main1}--\ref{fixup:phase3-main-check-end}. 

\noindent \underline {Steps~\ref{fixup:main1}--\ref{fixup:phase3-main-check-end}, Algorithm~\ref{algo:fixup}} --  As mentioned above, 
Steps~\ref{fixup:main1}--\ref{fixup:phase3-main-check-end}  create a set $S$
of triples that represent {\em new} shortest paths. There are two cases.
\begin{itemize}
\item  \underline{$P^{*}(x, y)$ is empty:}
Here, we process the triples in $S'$, but 
in addition,
we may be required to process triples of weight $wt$ from the set $P(x, y)$. 
To see this, consider the example in Fig.~\ref{fig:illust-example} and consider the pair $a_1, b_1$. 
In $G$, there is one
shortest path $\langle a_1, v, b_1 \rangle$ which is removed from $P(a_1, b_1)$ and $P^{*}(a_1, b_1)$ during cleanup. In $G'$,
$d(a_1, b_1) = 4$
and there are 2 shortest paths, namely $p_1 = \langle a_1, b_1 \rangle$
and $p_2 = \langle a_1, v_1, b_1\rangle$.  Note that both of these are LSPs in $G$ and 
therefore are present in $P(a_1,b_1)$. In Step~\ref{fixup:init2}, Algorithm~\ref{algo:fixup} we insert exactly one of them into
the heap $H_f$. However, both need to be processed and also left and right extended to create new
locally shortest paths. Thus, under this condition, we examine all the min-weight triples present in $P(a_1,b_1)$. 
\item  \underline{$P^{*}(x, y)$ is non-empty:} 
After a decremental update, the distance from $x$ to $y$ can either remain the same or increase,
but it cannot decrease.
Further,  cleanup removed from the \system all paths that contain $v$. Hence,
 if $P^{*}(x, y)$ is non-empty at this point,  it implies that all paths in $P^{*}(x, y)$ avoid
$v$.  
In this case, we can
show (Invariant~\ref{inv:fixup2}) that it suffices
to only examine the  triples present in $H_f$. Furthermore, the only paths that we need to process are the paths
that pass through the vertex $v$.
\end{itemize}
\noindent \underline {Steps~\ref{fixup:startleft}--\ref{fixup:right}, Algorithm~\ref{algo:fixup}} -- 
These steps left-extend and right-extend the triples in $S$ representing {\em new} shortest paths 
from $x$ to $y$.

\vspace{-0.2in}
\begin{algorithm}[H]
\scriptsize
\begin{algorithmic}[1]
\STATE $H_f \leftarrow \emptyset$; \MT $\leftarrow \emptyset$ \label{fixup:init0}
\FOR {each edge incident on $v$ } \label{fixup:init1}
	\STATE create a triple $\gamma$; set $paths(\gamma, v) = 1$; set update-num$(\gamma)$; add $\gamma$ to $H_f$ and to $P()$
\ENDFOR
\FOR {each $x, y \in V$}
	\STATE add a min-weight triple from $P(x, y)$ to $H_f$ \label{fixup:init2}
\ENDFOR
\WHILE {$H_f \neq \emptyset$} \label{fixup:phase3-begin}
        \STATE extract in $S'$ all triples with min-key $[wt,x,y]$ from $H_f$; $S \leftarrow \emptyset$ \label{fixup:phase3-extract1}
        \IF {$S'$ is the first extracted set from $H_f$ for $x,y$}  \label{fixup:phase3-first-ext}
		\STATE \COMMENT{Steps~\ref{fixup:main1}--\ref{fixup:phase3-add2LRStar1}: add new STs (or increase counts of existing STs) from $x$ to $y$.}
		\IF { $P^*(x,y)$ is empty} \label{fixup:main1}
        	\FOR {each $\gamma' \in P(x,y)$ with weight $wt$} \label{fixup:phase3-addfromP-begin}
                        	\STATE let $\gamma' = ((xa', b'y), wt, count')$
                        	\STATE add $\gamma'$ to $P^{*}(x,y)$ and $S$; add $x$ to $L^{*}(a',y)$ and $y$ to $R^{*}(x,b')$ \label{fixup:phase3-add2LRStar2}
	                \ENDFOR \label{fixup:phase3-addfromP-end}
        	\ELSE
        		\FOR {each $\gamma' \in S'$ containing a path through $v$}   \label{fixup:phase3-addfromX-begin}
       				\STATE let $\gamma' = ((xa', b'y), wt, count')$
                    \STATE add $\gamma'$ with $paths(\gamma',v)$ in $P^{*}(x,y)$ and $S$; add $x$ to $L^{*}(a',y)$ and $y$ to $R^{*}(x,b')$ \label{fixup:phase3-add2LRStar1}
        		\ENDFOR \label{fixup:phase3-addfromX-end}
        	\ENDIF \label{fixup:phase3-main-check-end}
        	\STATE \COMMENT{Steps~\ref{fixup:startleft}--\ref{fixup:endleft}: add new LSTs (or increase counts of existing LSTs) that extend SPs from $x$ to $y$.}
		\FOR {every $b$ such that $(x \times,by) \in S$} \label{fixup:startleft}
        		\STATE let $fcount'  = \sum_{i} ct_i$ such that $((xa_i, by), wt, ct_i ) \in S$
        		\FOR {every $x'$ in $L^{*}(x,b)$}
                		\IF {$(x'x, by) \notin$ \MT}
                        		\STATE $wt' \leftarrow wt+\weight(x',x)$; $\gamma' \leftarrow ((x'x,by),wt', fcount')$ \\
                     			\STATE set update-num$(\gamma')$; $paths(\gamma',v)\leftarrow \sum_{\gamma=(x\times,by)} paths(\gamma,v)$; add $\gamma'$ to $H_f$ \label{fixup:add1}
                        		\IF {a triple for $(x'x, by)$ exists in $P(x',y)$}
                                		\STATE add $\gamma'$ with $paths(\gamma', v)$ in $P(x',y)$; add $(x'x, by)$ to \MT 
                        		\ELSE
                                		\STATE add $\gamma'$ to $P(x',y)$; add $x'$ to $L(x,by)$ and $y$ to $R(x'x,b)$ 
                        		\ENDIF
                		\ENDIF
        		\ENDFOR
		\ENDFOR \label{fixup:endleft}
		\STATE perform steps symmetric to Steps  \ref{fixup:startleft} -- \ref{fixup:endleft} for right extensions.             \label{fixup:right}
        \ENDIF
\ENDWHILE \label{fixup:phase3-end}
\end{algorithmic}
\caption{fixup$(v, \weight')$}
\label{algo:fixup}
\end{algorithm}

Fixup maintains the following two invariants. The invariant below (Invariant~\ref{inv:fixup1}) shows that
for any pair $x, y$,
the weight of the first set of the triples extracted from $H_f$ determines the shortest path
distance from $x$ to $y$.  The proof of the invariant is similar to the proof of Invariant~3.1 in \cite{DI04}.

\begin{invariant}
\label{inv:fixup1}
If the set $S'$ in Step~\ref{fixup:phase3-extract1} of Algorithm~\ref{algo:fixup} is the first extracted set from
$H_f$ for $x,y$, then the
weight of each triple in  $S'$ 
is the shortest path distance
from $x$ to $y$ in the updated graph. 
\end{invariant}
\begin{proof}
Assume for the sake of contradiction that the invariant is violated
at some extraction. Thus, the first set of triples $S'$ of weight $\hat{wt}$ extracted for some pair $(x,y)$ do not represent the set of
shortest paths from $x$ to $y$ in the updated graph. Consider the earliest of these events
and let $\gamma = ((xa', b'y), wt, count)$ be a triple in the updated graph
that represents a set of shortest paths from $x$ to $y$
with $wt < \hat{wt}$. The triple $\gamma$ cannot be present in $H_f$,
else it would have been extracted before any triple of weight $\hat{wt}$
from $H_f$. Moreover, $\gamma$ cannot be in $P(x, y)$ at the beginning of
fixup otherwise $\gamma$ (or some other triple of weight $wt$) would
have been inserted into $H_f$ during Step~\ref{fixup:init2} of Algorithm~\ref{algo:fixup}.
Thus $\gamma$ must be a {\em new} LST generated by the algorithm.
Since all edges incident on $v$ are added to $H_f$ during Step~\ref{fixup:init1}
of Algorithm~\ref{algo:fixup} and $\gamma$ is not present in $H_f$, implies that
$\gamma$ represents paths which have at least two or more edges. We
now define left($\gamma$) as the set of LSTs of the form
$((xa, c_ib), wt - \weight(b, y), count_{i})$ that
represent all the LSPs in the left tuple $(xa, b)$; similarly we
define right$(\gamma)$ as the set of LSTs of the form $((ad_j, by), wt - \weight(x, a), count_j)$ that represent all the LSPs in the right
tuple $(a, by)$. Note that since $\gamma$ is a shortest path tuple,
all the paths represented by LSTs in left$(\gamma)$ and right$(\gamma)$ are
also shortest paths. All of the paths in either left$(\gamma)$ or  in  right$(\gamma)$
are {\em new} shortest paths and therefore are not present in $P^*$ at the beginning of fixup.
Since edge weights are positive $(wt - \weight(b,y)) < wt < \hat{wt}$ and
$(wt - \weight(x,a)) < wt < \hat{wt}$. As we extract paths from $H_f$ in
increasing order of weight, and all extractions before the wrong extraction
were correct, the triples in left$(\gamma)$ and right$(\gamma)$ should have been extracted from $H_f$
and added to $P^{*}$. Thus, the triple corresponding to $(xa, by)$ of
weight $wt$ should have been generated during left or right extension and
inserted in $H_f$. Hence, some triple of weight $wt$ must be extracted from $H_f$ for the pair
$(x,y)$ before any triple of weight $\hat{wt}$ is extracted from $H_f$. This
contradicts our assumption that the invariant is violated.
\qed
\end{proof}

Using Invariant~\ref{inv:fixup2} below we show that fixup indeed considers all of the {\em new} shortest paths for any pair $x,y$.
Recall that all the {\em new} shortest paths for a pair need not be present in $H_f$ and we may be required
to consider min-weight triples present in $P( \cdot)$ as well.

\begin{invariant}
\label{inv:fixup2}
The set $S$ of triples constructed in Steps~\ref{fixup:main1}--\ref{fixup:phase3-main-check-end} of Algorithm~\ref{algo:fixup}
represents all of the {\em new} shortest paths from $x$ to $y$.
\end{invariant}
\begin{proof}
Any new SP from $x$ to $y$ is of the following three types:
\begin{enumerate}
\item a single edge containing the vertex $v$ (such a path is
added to $P(x, y)$ and $H_f$ in Step~\ref{fixup:init1})
\item a path generated via left/right extension of some shortest path (such
a path is added to $P(x, y)$ and $H_f$ in Step~\ref{fixup:add1} and an analogous step in right-extend).
\item a path that was an LSP but not SP before the update and is an SP after the update.
\end{enumerate}
In (1) and (2) above
any new SP from $x$ to $y$  which is added to $H_f$ is also added to $P(x, y)$. However,
amongst the several triples representing paths of the form (3) listed above, only one candidate triple
will be present in $H_f$. Thus we conclude that for a given $x, y$ and when we extract from $H_f$
triples of weight $wt$, $P(x,y)$ contains
a superset of the triples that are present in $H_f$. We now consider the two cases that the algorithm deals with.
\begin{itemize}
\item $P^{*}(x, y)$ is empty when the first set of triples for $x, y$ is extracted from $H_f$. In this case,
we process all the min-weight triples in $P(x,y)$. By the above argument, we know that all new SPs from $x$ to $y$
are present in $P(x, y)$. Therefore it suffices to argue that all of them are {\em new}. Assume for the sake of contradiction,
some path $p$ represented by them is {\em not new}. By definition, $p$ does not contain $v$ and $p$ was a SP before
the update. Therefore, clearly $p$ was in $P^{*}(x,y)$ before the update. However, since cleanup only removes paths that contain $v$,
the path $p$ remains untouched during cleanup and hence continues to exist in $P^{*}(x,y)$. This contradicts the fact that $P^{*}(x,y)$
is empty.
\item $P^{*}(x,y)$ is not empty when the first set of triples for $x, y$ is extracted from $H_f$.
Let the weight of triples in $P^{*}(x,y)$ be $wt$.
This implies that the shortest path distance from $x$ to $y$ before and after the update is $wt$.
Recall that we are dealing with decremental updates.
We first argue that it suffices to consider triples in $H_f$. This is observed from the fact that
any {\em new} SP of the form (1) and (2) listed above is present in $H_f$. Furthermore,
note that any path of form (3) above has a weight strictly larger than $wt$ since it was an LSP and not SP before the update.
Thus in the presence of paths of weight $wt$, none of the paths of form (3) are candidates for shortest paths from $x$ to $y$.
This justifies considering triples only in $H_f$.

Finally, we note that for any triple considered, our algorithm only processes paths through $v$.
This again follows from the fact that only paths through $v$ were removed by cleanup and possibly need to be
restored if the distance via them remains unchanged after the update.
\end{itemize}
\qed
\end{proof}

The following lemma establishes the correctness of fixup.

\begin{lemma}
\label{lem:fixup-corr1}
After execution of Algorithm~\ref{algo:fixup}, for any $(x, y) \in V$, the counts of the triples
in $P(x,y)$
and $P^{*}(x,y)$ represent the counts of LSPs and SPs from $x$ to $y$ in the updated graph.
Moreover,
the sets $L, L^{*}, R, R^{*}$ are correctly maintained.
\end{lemma}
\begin{proof}
We prove the lemma statement by showing the invariants are maintained by the while loop in Step~\ref{fixup:phase3-begin} of Algorithm~\ref{algo:fixup}. \\
{\noindent \bf Loop Invariant:} At the start of each iteration of the while loop in Step~\ref{fixup:phase3-begin} of Algorithm~\ref{algo:fixup}
let the min-key triple to be extracted and processed from $H_f$ have key = $[wt, x, y]$.
We claim the following about the \system and $H_f$.

\begin{enumerate}
\item [$\mathcal{I}_1$] \label{proof:fitem1} For any $a, b \in V$, if $G'$ contains $c_{ab}$ number of LSPs of weight ${wt}$ of the form $(xa, by)$.
Further, a triple $\gamma = ((xa, by), { wt}, c_{ab})$ is present in $P(x,y)$ (note that $H_f$ can also contain other triples from $x$ to $y$ with weight $wt$).

\item [$\mathcal{I}_2$]  Let $[\hat{wt}, \hat{x}, \hat{y}]$ be the last key extracted from $H_f$ and processed before $[wt,x,y]$. For any key $[wt_1, x_1, y_1] \leq [\hat{wt}, \hat{x}, \hat{y}]$, let
$G'$ contain ${ c}  > 0 $ number of SPs of weight  ${ wt_1}$ of the
form $(x_1a_1, b_1y_1)$. Further, let ${ c_{new}}$ (resp. ${ c_{old }}$) denote the number of such SPs
 that are {\em new} (resp. not {\em new}).
Here ${ c_{new} + c_{old} = c}$. Then,
\begin{enumerate}
\item  \label{proof:fitem2} 
the triple for $(x_1a_1, b_1y_1)$ with weight ${wt_1}$ in $P^{*}(x_1, y_1)$ represents $c$ SPs.
\item \label{proof:fitem3} $x_1 \in L(a_1, b_1y_1)$, $x_1 \in L^{*}(a_1, y_1)$,  and $y_1 \in R(x_1a_1, b_1)$, $y_1 \in R^{*}(x_1, b_1)$.
Further, $(x_1a_1, b_1y_1) \in $ \MT
iff ${ {c_{old}} > 0}$.
\item \label{proof:fitem4} If $c_{new} > 0$, for every $x' \in L(x_1, b_1y_1)$, a triple corresponding to $(x' x_1, b_1y_1)$
with weight $wt'~=~wt_1~+~\weight (x'x_1)$ and the appropriate count
is in $P(x_1,y_1)$ and in $H_f$
if $[wt', x', y_1] \geq [wt, x, y]$. A similar claim can be stated for every $y' \in R(x_1a_1, y_1)$.
\end{enumerate}
\item [$\mathcal{I}_3$] \label{proof:fitem5} For any key $[wt_2, x_2, y_2 ] \geq [wt, x, y]$, let
$G'$ contain $c > 0$ number of LSPs of weight  ${ wt_2}$ of the
form $(x_2a_2, b_2y_2)$. Further, let ${ c_{new}}$ (resp. ${ c_{old }}$) denote the number of such SPs
 that are {\em new} (resp. not {\em new}).
Here ${ c_{new} + c_{old} = c}$. Then the tuple $(x_2 a_2, b_2 y_2) \in$ \MTend, iff
$c_{old} > 0$ and $c_{new}$ paths have been added to $H_f$ by some earlier iteration of the while loop.
\end{enumerate}

The proof that these invariants hold at initialization and termination and are maintained at every iteration
of the while loop is similar to the proof of Lemma~\ref{lem:cleanup-corr2}.
\qed
\end{proof}

\subsubsection{Complexity of Fixup.} 
\label{sec:complexity}
As in DI, we observe that shortest paths and LSPs are removed only in cleanup and are added only in fixup.
In a call to fixup, accessing a triple takes $O(\log n)$ time
since it is accessed on a constant number of data structures. So, it suffices to bound the number of triples accessed in a call to fixup, and then multiply that bound by  $O(\log n)$.

We will establish an amortized bound.
The total number of LSTs at any time, 
including the end of the update sequence,
is $O(m^* \cdot \vstar)$ (by Lemma~\ref{lem:total-locally-shortest}).
Hence, if fixup accessed only {\em new} triples outside of the $O(n^2)$ triples added initially to $H_f$,  the 
amortized cost of fixup (for a 
sufficiently long update sequence) would 
be $O(\vstar^2 \cdot \log n)$, the cost of a cleanup. This is in fact the analysis in DI, where fixup satisfies this property. However,
in our algorithm fixup accesses several  triples that are already in the tuple system: In
Steps~\ref{fixup:phase3-addfromP-begin}--\ref{fixup:phase3-addfromP-end} we examine triples already in $P$, 
in Steps~\ref{fixup:phase3-addfromX-begin}--\ref{fixup:phase3-addfromX-end} we could increment the count of an
existing triple in $P^*$,  and in Steps~\ref{fixup:startleft}--\ref{fixup:endleft} we increment the count of an existing triple in $P$.
We bound the costs of these steps in Lemma~\ref{lem:fixup-time} below by classifying each triple $\gamma$ as one of the
following disjoint types:

\begin{itemize}
\item {\bf Type-0 (contains-v):} $\gamma$  represents  at least one path containing vertex $v$.
\item {\bf Type-1 (new-LST):} $\gamma$ was not an LST before the update but is an LST after the update,
and no path in $\gamma$ contains $v$.
\item {\bf Type-2 (new-ST-old-LST):} 
 $\gamma$ is an ST after the update,
and $\gamma$ was an LST but not an ST before the update, and  no path in $\gamma$ contains $v$.
\item {\bf Type-3 (new-ST-old-ST):} $\gamma$  was an ST  before the update and continues
to be an ST after the update,
and no path in $\gamma$ contains $v$.
\item {\bf Type-4 (new-LST-old-LST):} $\gamma$  was an LST  before the update and continues
to be an LST after the update,
and no path in $\gamma$ contains $v$.
\end{itemize}

The following lemma establishes an amortized bound for fixup which is the same as the worst case bound for cleanup.  

\begin{lemma}
\label{lem:fixup-time}
The fixup procedure takes time $O(\vstar^2 \cdot \log{n})$ amortized over a sequence of $\Omega(\maxM / \vstar)$ decremental-only updates.
\end{lemma}

\begin{proof} 
We bound the number of triples examined; the time taken is $O(\log n)$ 
times the number of triples
examined due to the data structure operations performed on a triple. The initialization in Steps~\ref{fixup:init0}--\ref{fixup:init2} takes $O(n^2)$ time. We now consider the triples examined after Step~\ref{fixup:init2}.
The number of  Type-0 triples is $O(\vstar^2)$ by Lemma~\ref{lem:bound-tuples-thru-v}.
The number of Type-1 triples is addressed by amortizing over the entire update sequence as described in the paragraph below. For \mbox{Type-2} triples we observe that since updates only increase 
the weights on edges, a shortest path never reverts to being an LSP.
Further, each such Type-2 triple is examined only a constant number of times (in Steps~\ref{fixup:main1}--\ref{fixup:phase3-add2LRStar2}). Hence we charge each access to a Type-2 triple to the step in which it was created as a Type-1 triple. 
For Type-3 and Type-4, we note that for any $x, y$ we add exactly one candidate
min-weight triple from $P(x,y)$ to $H_f$,
hence initially there are at most $n^2$ such triples in $H_f$.
Moreover, we never process an old
LST which is not an ST so no additional Type-4 triples are examined during fixup. Finally,
triples in $P^*$ that are not 
placed initially in $H_f$ are not examined in any step of fixup,
so no additional Type-3 triples are examined. 
Thus the number of triples examined by a call to fixup is
$O(\vstar^2)$ plus $O(X)$, where $X$ is the number of {\em new} triples fixup adds to the tuple system.
(This includes an $O(1)$ credit placed on each new LST for a possible later conversion to an ST.)
  
Let $\sigma$ be the number of updates in the update sequence.  Since triples are removed only in cleanup, at most $O(\sigma \cdot \vstar^2)$ triples are removed by the cleanups. 
There can be at most $O(m^* \cdot \vstar)$ triples remaining at the end of the sequence (by Lemma 1), hence the total number of new triples added by all fixups in the update sequence is
$O(\sigma \cdot \vstar^2 + m^* \cdot \vstar)$. When $\sigma > m^*/\vstar$, the first term dominates, 
and this gives an average of $O( \vstar^2)$ triples added per fixup, and the desired amortized time bound for fixup.
\qed
\end{proof}

\subsection{Complexity of the Decremental Algorithm.}
Lemma~\ref{lem:fixup-time} establishes that the amortized cost per update of fixup is $O(\vstar^2 \cdot \log{n})$ 
when the decremental update sequence is of length $\Omega(\maxM / \vstar)$. Lemma~\ref{lem:cleanup-time}
shows that the worst case cost per update of cleanup is  $O(\vstar^2 \cdot \log{n})$. Since an update
operation consists of a call to cleanup followed by a call to fixup, this establishes Theorem~\ref{th:main}.

\section{Discussion}

We have presented an efficient decremental algorithm to maintain all-pairs all shortest paths (APASP).
The space used by our algorithm is $O(\maxM \cdot \vstar)$, 
the worst case  number of triples in our tuple system.
By using this
decremental APASP algorithm in place of the incremental APASP algorithm used in~\cite{NPR14},
we obtain a decremental algorithm 
with the same bound for maintaining BC scores.

Very recently, two of the authors have obtained a fully dynamic APASP algorithm~\cite{PR14}
that combines elements in the
fully dynamic APSP algorithms in ~\cite{DI04} and \cite{Thorup04}, while building on the results in the current paper. When specialized to unique shortest paths (i.e., APSP), this algorithm is about as simple as the one in~\cite{DI04}
and matches its amortized bound.

\bibliographystyle{abbrv}
\bibliography{references,refs2}
\newpage
\begin{appendix}

\section{Details from Section~\ref{sec:cleanup}}
\subsection{Accumulation}
\label{app:accum}
\noindent {\bf Need for accumulation.} In Step~\ref{process-cleanup:left-extend-start} of
Algorithm~\ref{algo:cleanup}, we consider every $b$ such that
$(x \times, by)$ belongs to $S$. We also assume that we have the
accumulated count of such triples
available in Step~\ref{process-cleanup:left-extend-fcount}.
An efficient method to accumulate these
counts is given below. We use this
accumulated count to generate a longer LST for each $x' \in L(x, by)$
(Step~\ref{process-cleanup:create-left}). (For the moment, ignore
the check of a tuple being present in \MTend.)
Consider our example in Fig.~\ref{fig:illust-example} where after the decremental update on $v$, we intend to remove
from the tuple system the following two paths passing through $v$ namely (i) $p_1 = \langle x' x, a_1, v, b, y \rangle$
and (ii) $p_2 = \langle x', x, a_2, v, b, y \rangle$. Note that both these paths represented by triples of the form  $(x'x, by)$.
We further remark that $p_1$ can be generated by left extending the triple $((xa_1, by), 4, 1)$
whereas $p_2$ can be generated by left extending the triple $((xa_2, by), 4, 1)$. 
However, instead of left-extending each triple individually, our algorithm
accumulates the count to obtain $2$ paths represented by the $r$-tuple $(x, by)$ and then generates the triple $((x'x, by), 5, 1)$. We note that
such
an implementation is correct because any valid left extension of triples of the form $(xa_1, by)$
is also a valid left extension of triples of the form $(xa_2, by)$ when the triples have the same weight.
Furthermore, it is efficient since it generates the triple of the form $(x'x, by)$ at most once.
This is the precise reason for defining the set $L$ for an $r$-tuple $(x, by)$ instead of defining
it for the tuple $(xa_1, by)$.\\

\noindent{\bf Accumulation technique.}
An efficient implementation of getting accumulated counts can be achieved in several ways. For the sake of concreteness, we
sketch an implementation by maintaining two arrays $A$ and $B$ of size $n$ each and two linked
lists $L_a$ and $L_b$. Assume that the
arrays are initialized to zero and the linked lists are empty just before any triple with key $[wt, x, y]$  is extracted from the heap. When a triple
$\gamma = ((xa_i, b_jy), wt, count_{ij})$ is extracted from $H_c$, we add $count_{ij}$ to
$A[a_i]$ and $B[b_j]$. The lists $L_a$ and $L_b$ maintain pointers to non-zero locations in the arrays $A$ and $B$ respectively.
Thus, when all triples of weight $wt$ corresponding to tuples of the form $(x\times, \times y)$ are
extracted from $H_c$, the value in $A[a_i]$ denotes the  number of locally shortest
paths of the form $(xa_i, \times y)$ to be updated. Similarly, the value in $B[b_j]$ denotes
 number of locally shortest
paths of the form $(x\times, b_j y)$ to be updated. Using the lists $L_a$ and $L_b$,
we can efficiently access the accumulated counts as well as reinitialize (to zero) all the non-zero values in the two arrays $A$ and $B$.

\subsection{Need for Marked-Tuples}
\label{app:marked}
Consider the example in Fig.~\ref{fig:illust-example}
and assume that we have deleted the two paths of the form $(x'x, by)$ which pass through $v$.
Furthermore, assume that we have generated them via left extending the two triples of the form
$(xa_1, by)$ and $(xa_2, by)$. Now note that since path $\langle x', x, a_2, v_2, b, y \rangle$
continues to exist in the tuple system, $x' \in L(x, by)$ and $y \in R(x'x, b)$. Thus, when we consider the triples of the form $(x'x, b)$
for right extension, it is possible to generate the same paths again. To avoid such a
double generation we use the dictionary \MTend. In Step~\ref{process-cleanup:left-extend-unmarked} of Algorithm~\ref{algo:cleanup},
just before we create a left extension of a set of triples of the form $(x \times, by)$
using the vertex $x' \in L(x, by)$, we check whether $(x'x, by)$ is present in \MTend.
Recall that, \MT is empty when the cleanup begins.
When  a triple for  $(x'x, by)$  is generated for the first time (either by a left extension
or right extension),
and there are additional locally shortest paths in $G$ of the form $(x'x, by)$ which do not pass through $v$,
we insert a tuple $(x'x, by)$ in \MT (Step~\ref{process-cleanup:left-extend-markR}, Algorithm~\ref{algo:cleanup}).
Thus the data structure \MT and the checks in Step~\ref{process-cleanup:left-extend-unmarked} of Algorithm~\ref{algo:cleanup}
ensure that the paths are generated exactly once either as a left extension or as a right extension but not by both.
Note that such a marking is not required when
there are no additional paths in $G$ which do not pass through $v$. In that case, we immediately delete
$x'$ from $L(x, by)$ and $y$ from $R(x'x, b)$ (Step~\ref{process-cleanup:left-extend-removeL}, Algorithm~\ref{algo:cleanup})
ensuring that a triple for $(x'x, by)$ gets generated exactly once.
This is the only case that can occur in
DI~\cite{DI04} due to the assumption of unique shortest paths,
 and therefore this book-keeping with \MT
is not required in \cite{DI04}.

\subsection{Analysis of Cleanup}
\label{app:corr1}
\begin{appendix-lemma}
{\ref{lem:cleanup-corr2}}
After Algorithm~\ref{algo:cleanup}
is executed, the
counts of triples in $P$ ($P^{*}$) represent counts of LSPs (SPs)  in $G$
that do not pass through $v$. Moreover, the sets $L, L^{*}, R, R^{*}$ are correctly maintained.
\end{appendix-lemma}

\begin{proof}
To prove the lemma statement we show that the while loop in Step~\ref{cleanup:while}
of Algorithm~\ref{algo:cleanup} maintains the following invariants.\\
{\bf Loop Invariant:} At the start of each iteration of the while loop in Step~\ref{cleanup:while}
of Algorithm~\ref{algo:cleanup}, assume that the min-key triple to be extracted and processed from $H_c$ has key $[wt, x, y]$. Then
the following properties hold about the tuple system and $H_c$.
We assert the invariants about the sets $P$, $L$, and $R$. Similar arguments can be used to establish the
correctness of the sets $P^{*}$, $L^{*}$, and $R^{*}$.

\begin{enumerate}
\item [$\mathcal{I}_1$]\label{proof:item1} For any $a, b \in V$, if $G$ contains $c_{ab}$ number of locally shortest paths of weight ${wt}$ of the form $(xa, by)$
passing through $v$,
then $H_c$ contains a triple $\gamma = ((xa, by), { wt}, c_{ab})$. Further, $c_{ab}$ has been
decremented from the initial count in the triple for $(xa, by)$  in $P(x,y)$.
\item [$\mathcal{I}_2$] Let $[\hat{wt}, \hat{x}, \hat{y}]$ be the key extracted from $H_c$ and processed in the previous iteration.
For any key $[wt_1, x_1, y_1] \leq [\hat{wt}, \hat{x}, \hat{y}]$, let
$G$ contain ${ c}  > 0 $ number of LSPs of weight  ${ wt_1}$ of the
form $(x_1a_1, b_1y_1)$. Further, let ${ c_v}$ (resp. ${ c_{\bar v}}$) denote the number of such LSPs
 that pass through $v$ (resp. do not pass through $v$).
Here ${ c_v + c_{\bar v} = c}$. Then,
\begin{enumerate}
\item \label{proof:item2} if $c>c_v$ there is a triple in $P(x_1,y_1)$ of the form $(x_1a_1, b_1y_1)$  and weight $wt_1$ representing $c-c_v$ LSPs. If $c=c_v$ there is no such triple in $P(x_1,y_1)$.
\item \label{proof:item3} $x_1 \in L(a_1, b_1y_1)$,  $y_1 \in R(x_1a_1, b_1)$, and    $(x_1a_1, b_1y_1) \in $ \MT
iff ${ {c_{\bar v}} > 0}$.
\item \label{proof:item4} For every $x' \in L(x_1, b_1y_1)$, a triple corresponding to $(x' x_1, b_1y_1)$
with weight $wt'~=~wt_1~+~\weight (x',x_1)$ and the appropriate count is 
in $H_c$ if $[wt', x', y_1] \geq [wt, x, y]$. A similar claim can be stated for every $y' \in R(x_1a_1, y_1)$.
\end{enumerate}
\item [$\mathcal{I}_3$] \label{proof:item5} For any key $[wt_2, x_2, y_2 ] \geq [wt, x, y]$, let
$G$ contain $c > 0$ LSPs of weight  ${ wt_2}$ of the
form $(x_2a_2, b_2y_2)$.
Further, let ${ c_v}$ (resp. ${ c_{\bar v}}$) denote the number of such LSPs
 that pass through $v$ (resp. do not pass through $v$).
Here ${ c_v + c_{\bar v} = c}$.
Then the tuple $(x_2 a_2, b_2 y_2) \in$ \MTend, iff $c_{\bar v} > 0$ and a triple for
$(x_2a_2, b_2 y_2)$ representing $c_v$ LSPs is present in $H_c$.
\end{enumerate}

\noindent {\bf Initialization:}
We show that the invariants hold at the start of the first
iteration of the while loop in Step~\ref{cleanup:while} of Algorithm~\ref{algo:cleanup}.
The min-key triple in $H_c$ has key $[0, v, v]$. Invariant~{$\mathcal{I}_1$}
holds since  we inserted into
$H_c$ the trivial triple of weight $0$ corresponding to the vertex $v$
and that is the only triple of such key. Moreover, since we do not represent trivial
paths containing the single vertex, no counts need to be decremented.
Since we assume positive edge weights, there are no LSPs
in $G$ of weight less than zero. Thus, invariants $\mathcal{I}_2$(a), $\mathcal{I}_2$(b), and $\mathcal{I}_2$(c)
hold trivially.
Invariant~$\mathcal{I}_3$ holds since $H_c$ does not contain any triple of weight $> 0$ and we initialized \MT to empty.

\noindent {\bf Maintenance:} Assume that the invariants are true at the beginning of the $1 \le i \le k$-th iteration
of the while loop.
We now prove that the claims are true at the beginning of the $(k+1)$-th iteration.
Let the min-key triple at the beginning of the $k$-th iteration be $[wt_k, x_k, y_k]$.
By invariant~$\mathcal{I}_1$, we know that for any $a_i, b_j$, if there exists LSPs in $G$ of the form $(x_k a_i, b_j y_k)$
of weight $wt_k$, they have been inserted into $H_c$ and further their counts have been
decremented from appropriate triples in $P(x_k,y_k)$.
Now consider the set of triples with key $[wt_k, x_k, y_k]$
which we extract in the set $S$ (Step~\ref{cleanup:extract-set}, Algorithm~\ref{algo:cleanup}).
We consider left-extensions of triples in $S$; symmetric arguments apply for right-extensions.
Consider for a particular~$b$, the set of triples $S_{b} \subseteq S$
and let $fcount'$ denote the sum of the counts of the paths represented by triples in $S_b$.
Let $x' \in L(x_k, by_k)$; our goal is to generate the paths $(x'x_k, b y_k)$ with
count = $fcount'$ and weight $wt' = wt_k+\weight(x',x_k)$.
 However, we generate such paths
only if they have not been generated by a right-extension of another set of paths.
We note that the paths
of the form $(x'x_k, by_k)$ can be generated by right extending the set of
triples of the form $(x'x_k, \times b)$. Without loss of generality assume that the triples of
the form $(x'x_k, \times b)$ have a key which is greater than the key $[wt_k, x_k, y_k]$ and they are not in $H_c$. Thus,
at the beginning of the $k$-th iteration, by invariant~$\mathcal{I}_3$, we know that $(x'x_k, by_k) \notin$ \MTend.
Steps~\ref{process-cleanup:create-left}--\ref{process-cleanup:left-extend-removeP},
Algorithm~\ref{algo:cleanup} create a triple of the form $(x'x_k, by_k)$ of weight $wt'$ and decrement
$fcount'$ many paths from the appropriate triple in $P(x',y_k)$ and add it to $H_c$.
This establishes invariants~$\mathcal{I}_2$(a) and $\mathcal{I}_2$(c) at the beginning of the $(k+1)$-th iteration.
In addition, if there are no
LSPs in $G$ of the form $(x'x_k, by_k)$ which do not pass through $v$,
we delete $x'$ from $L(x_k, by_k)$ and delete $y_k$ from $R(x'x_k, b)$ (Step~\ref{process-cleanup:left-extend-removeL}, Algorithm~\ref{algo:cleanup}).
On the other hand, if there exist LSPs in $G$ of the form $(x'x_k, by_k)$,
then $x'$ (resp. $y_k$) continues to exist in $L(x_k, by_k)$ (resp. in $R(x'x, b)$). Further,
we add the tuple $(x'x_k, by_k)$ to \MT and note that the corresponding triple is already present in $H_c$ (Step~\ref{process-cleanup:left-extend-markR}, Algorithm~\ref{algo:cleanup}).
Since the invariants
$\mathcal{I}_2$(b)  and $\mathcal{I}_2$(c) were true for every key  $ < [wt_i, x_i, y_i]$
and by the above steps we ensure that these invariants hold for every key = $[wt_i, x_i, y_i]$.
Thus,  invariant
$\mathcal{I}_2$(b)  
is true at the beginning of the $(k+1)$-th iteration.
Note that any triple that is generated by a left extension (or symmetrically
right extension) is 
inserted into $H_c$ as well as into \MTend. This establishes invariant~$\mathcal{I}_3$ at the beginning of the $(k+1)$-th iteration.

Finally, to see that invariant~$\mathcal{I}_1$ holds at the beginning of the $(k+1)$-th iteration, let the
min-key at the $(k+1)$-th iteration be $[wt_{k+1}, x_{k+1}, y_{k+1}]$. Note that triples
with weight $wt_{k+1}$ starting with $x_{k+1}$ and ending in $y_{k+1}$ can be created
either by left extending or right extending the triples of smaller weight. And since for each of
iteration $\le k$ invariant $\mathcal{I}_2$(c) holds, we conclude that invariant~$\mathcal{I}_1$ holds at the beginning of the $(k+1)$-th iteration.

\noindent {\bf Termination:} The exit condition of the while loop is when the heap $H_c$ is empty. Because Invariant $\mathcal{I}_1$ maintains in $H_c$ the first triple to be extracted and processed, then $H_c=\emptyset$ implies that there
are no more triples containing the vertex $v$ that need to be left or right extended and removed from the tuple system. 
Moreover, since the invariants hold for the last set of triples of weight $\hat{wt}$ extracted from the heap, by $\mathcal{I}_2$(a),
all LSPs having weight less than or equal to $\hat{wt}$ have been decremented from the appropriate sets $P(\cdot)$. Finally, 
due to $\mathcal{I}_2$(b), the sets $L$ and $R$ are also correctly maintained after the while loop terminates.
\qed
\end{proof}

\noindent
In an analogous way we can prove the correctness of the loop invariant for fixup given in the proof of Lemma~\ref{lem:fixup-corr1}.
\end{appendix}

\end{document}